\documentclass{article}
\usepackage{graphicx} 

\usepackage{amsmath}
\usepackage{amssymb} 
\usepackage{comment}
\usepackage{verbatim}
\usepackage{caption}
\usepackage{subcaption}

\usepackage{tikz}

\usepackage{authblk}
\usepackage{amsthm}
\usepackage{verbatim}
\usepackage{booktabs}

\newtheorem{prop}{Proposition}
\newtheorem{lemma}{Lemma}
\newtheorem{Corollary}{Corollary}
\newtheorem{thrm}{Theorem}
\newcommand{\Xomit}[1]{}

\newtheorem{example}{Example}

\usepackage{dsfont}

\usepackage[
natbibapa
]{apacite} 
\usepackage[hyphens]{url}
\begin{document}
\title{Prophet Inequalities via Linear Programming}
\author[1]{Halil I. Bayrak}
\author[2]{Mustafa \c{C}. P{\i}nar}
\author[3]{Rakesh Vohra}
\affil[1]{Technical University of Munich}
\affil[2]{Bilkent University}
\affil[3]{University of Pennsylvania}


\date{February 2026}

\maketitle

\section*{Abstract}
Prophet inequalities bound the expected reward that can be obtained in a stopping problem by the optimal reward of its corresponding off-line version. We propose a systematic technique for deriving prophet inequalities for stopping problems associated with selecting a point in a polyhedron. It utilizes a reduced-form linear programming representation of the stopping problem. We illustrate the technique to derive a number of known results as well as some new ones. For instance, we prove a $\frac{1}{2}$-prophet inequality when the underlying polyhedron is an on-line polymatroid; one whose underlying submodular function depends upon the realized rewards. 
We also demonstrate a composition by the Minkowski sum property. If an $r-$ prophet inequality holds for polyhedra $P^1$ and $P^2$, it also holds for their Minkowski sum.

\newpage

\section{Introduction}

Optimal stopping problems are concerned with choosing a time to
take a given action based on sequentially observed random variables so as to maximize
an expected payoff.  The action taken may be to reject a hypothesis, replace a
machine, or hire a secretary. They can be solved using dynamic programming. 

The basic stopping problem is the archetypal on-line resource allocation problem. There is one unit of an indivisible resource to be allocated. Requests for the resource arrive sequentially, and associated with each request is a reward for honoring it. The challenge is to decide which request is to be honored {\em before} knowing what the rewards from future requests will be. It is straightforward to generalize this basic stopping problem to cases where a request needs multiple heterogeneous resources.

A prophet inequality lower bounds the ratio
of the optimal expected value of the on-line problem to the optimal expected value of its off-line counterpart. The off-line counter part is often called the prophet's problem, because the decision maker has knowledge of {\em all} rewards. The larger the bound, the less valuable perfect information about the future is. This ratio also represents the limit of what any approximation algorithm for the on-line problem can achieve (in a worst-case sense).

In the `prophet problem', introduced in \citet{krengel1977semiamarts}, one is shown $n$ non-negative rewards, sequentially, that are independent draws from known distributions (not necessarily identical). After observing each reward, one can reject it and see the next reward in the sequence, or accept it, in which case the process stops. Note, one cannot select any previously rejected reward. The goal is to make a selection to maximize the expected reward and compare it to the expected reward one can obtain with hindsight. The maximum expected reward in hindsight is the expected value of the $n^{th}$ order statistic. In \citet{krengel1977semiamarts} a tight bound of 1/2 is obtained. In other words, the optimal expected reward of the stopping problem is at least half the size of the expected value of the $n^{th}$ order statistic. 

We consider an extension that can be interpreted as
selecting a point in a polyhedron where the value of each coordinate must be chosen sequentially. The sequence of coordinates to be chosen can be interpreted as a sequence of service requests, with the value of the coordinate being a service level. Associated with the $i^{th}$ coordinate or service request is a random number,  $r_i$, to be interpreted as the marginal value of honoring the $i^{th}$ request. Each service request may consume a vector of resources, and the underlying polyhedron characterizes which combinations of resources are feasible to consume. We allow for a `fractional' selection of each reward, which presumes that the resources being consumed are divisible. A fractional selection constrained to lie in $[0,1]$ can be interpreted as a probability of selection.  

The problem bears obvious similarities to the on-line linear programming problem studied, for example, in \citet{AgrawalWangYe2014}. The difference is that in our setting, the feasible region, as defined by a collection of inequalities, say, $Ax \leq b$, is known in advance, but not the objective function coefficients of the variables. In the on-line linear programming problem, the constraint matrix $A$ is revealed one column at a time. Furthermore, the bounds derived in the on-line linear programming case are asymptotic in nature, i.e., depend on how the size of the problem grows.

We describe one setting where the assumption that the feasible region is known in advance, but not the rewards, is realistic: computational sprinting. Recent chip micro-architecture developments make it possible
to expedite processes running in the processor. The technique is called sprinting, and it is
a class of mechanisms that provides a short but significant performance boost while temporarily
exceeding the thermal design point of the processors. \citet{HuangJoaoRicoHiltonLee2019} propose an algorithm that manages sprints by dynamically predicting utility and modeling thermal
headroom. The authors also compare their set of algorithms experimentally against
an {\em oracular policy}, which matches the notion of a prophet in our setting.

Each epoch consists of a large number of instructions and can be categorized accordingly.
When the computation reaches an epoch, the algorithm can predict the
gains from sprinting accurately. The algorithm has to decide the extent of increasing the temperature
beyond the thermal design point in each epoch. Apart from the thermal headroom available
in each period, there is a limit on the number of periods that the processor can perform
over the thermal design point. The inter-temporal limitation on heat increase can be formulated as
a knapsack constraint. An allocation for the fractional knapsack constraint serves as
a guide to pace sprints to maximize long-run performance under thermal constraints. In reality, limits on heat increases are not the only constraints. There are limitations on voltage and power as well.

In this paper, we use a reduced-form linear programming representation to obtain prophet inequalities for this class of prophet problems. Such reduced form representations have been useful in stochastic scheduling, mechanism, and information design (see \citet{Bhattacharya1992ExtendedPP}, \citet{bertnino} \citet{che2013} and \citet{vtv}). As demonstrated here, the approach gives a simple technique for deriving prophet inequalities, particularly when contrasted with the leading alternative called the on-line contention resolution scheme (see \citet{ocrs}). 

Linear programming has been used to attack stopping problems before. \citet{dk} is an early example. The approach described here (motivated by \citet{epitropou2019optimal}) is different in that the distinction between the on-line and off-line problem is in the feasible region and {\em not} the decision variables. This allows one to take the optimal solution to the off-line problem and rescale it into a feasible solution for the on-line problem. We defer a discussion of the precise difference to the end of Section \ref{section: model}.

Our approach yields, as we illustrate, a {\em systematic} way to derive known prophet inequalities for a variety of different stopping problems. Specifically,  
\begin{itemize}
\item[(i)] A $\frac{1}{K+1}$-prophet inequality when the polyhedron has $K$ many constraints and rewards are drawn independently.
	\item[(ii)] A $\frac{1}{2}$-prophet inequality under a polymatroid constraint and independently drawn rewards.
	\item[(iii)] $\frac{1}{n}$-prophet inequality when the $n$ fractional selections are constrained to lie in a polyhedron and the rewards are not necessarily independent.
\end{itemize}

We also provide some new results. The first is a composition by the Minkowski sum property. If an $r-$ prophet inequality holds for polyhedra $P^1$ and $P^2$, it also holds for their Minkowski sum. Second, we derive a $\frac{1}{2}$-prophet inequality for the case where the underlying polyhedron is a polymatroid that {\em depends} on the realized rewards. In other words, the available capacity of each resource is a function of realized rewards. High realized rewards (e.g., prices or willingness to pay) may justify turning on extra servers,
bringing peaker plants on-line, or expanding capacity. More generally, systems where supply replenishment depends on realized rewards would fit. 

Relevant prior work is summarized next, the model and notation in the subsequent section followed by the applications.

\textit{Literature Review.} The literature on prophet inequalities is vast. The interested reader can consult \citet{hill1992survey} for a classic survey. \citet{lucier2017economic} provides a more recent survey focusing on the connections to mechanism design. More relevant for us are
\citet{alaei2014bayesian}, \citet{kleinberg2012matroid} and \citet{dutting2015polymatroid}. The first allows for a selection from a uniform matroid, the second from arbitrary matroids, and the last, which subsumes the previous two, allows selections from polymatroids.
Our method replicates their result, $\frac{1}{2}-$prophet inequality for polymatroids, via simple linear programming tools and extends it to on-line polymatroids.\footnote{Our result is a little stronger in that we do not require the associated submodular function to be integer valued.}

There is a literature on the use of linear programming to solve stopping problems (see, for example, \citet{chostock}). However, the focus is on continuous-time settings and computational speed rather than on deriving prophet inequalities. In discrete time and discrete space settings, linear programming models were used for modeling optimal stopping to price American contingent claims in financial markets (see, e.g., \citet{camcimcp}).

\section{Model} \label{section: model}
We consider an extension of the basic stopping problem that can be interpreted as
selecting a point in a polyhedron where the value of each coordinate must be chosen sequentially. We refer to this as the {\bf on-line selection problem}.
The polyhedron is described by a collection of inequalities, e.g.,
knapsack or polymatroid. 

We let $n$ denote a natural number, $n \in \mathbb{N}$, which will be the total number of rewards. We use $[n]$ to denote the set $\{1,2,\ldots,n\}$. Associated with the $i^{th}$ coordinate or service request is a random number,  $r_i$, to be interpreted as the marginal value of honoring the $i^{th}$ request. Each $r_i$ is drawn from a known distribution $F_i$ (with density $f_i$) over a finite set $R_i$. The provision of services requires a bundle of resources. To provide one unit of service to the $i^{th}$ request requires $a_{ki}$ units of the $k^{th}$ resource, and $b_k$ is the available supply of resource $k \in [K]$, where $K$ is the total number of resources/constraints.

We use bold text when denoting vectors, thus, $\mathbf{r}$ denotes an $n$-vector of rewards $(r_1,r_2,\ldots,r_n)\in \times_{\ell \in [n]} R_\ell$. The truncated vector, $(r_1,r_2,\ldots,r_i) \in \times_{\ell \in [i]} R_\ell$ for any $i \in [n]$ will is denoted $\mathbf{r}^i$. We sometimes write $\mathbf{r}$ as $(r_j,\mathbf{r}_{-j})$, where $\mathbf{r}_{-j}$ is the vector of rewards obtained from $\mathbf{r}$ by removing reward $r_j$.

 We first define the variables for the on-line selection problem. Let $q_j(\mathbf r^j)$ be the service  level selected for request $j \in [n]$, given that the profile of rewards $\mathbf r^j {\in \times_{\ell \in [j]} R_\ell}$ was realized. If we require $q_j(\mathbf r^j) \in [0,1]$ for all $j$ and all profiles $\mathbf r^j$, we can interpret it as a probability. Thus, we are allowing for randomized selection. 

As in the mechanism design literature, we refer to $\mathbf q$ as an {\bf ex-post allocation} variable. For all $j \in [n]$, let
\[{{\mathbf{q}}}(\mathbf r^i) = [q_1(r_1), q_2(\mathbf{r}^2), \ldots, q_j(\mathbf r^j)]. \]
To describe the set of feasible ex-post allocations, let $A$ be the non-negative $K \times n$ matrix whose $(k,j)^{th}$ entry is $a_{kj}$. Let $A^j$ denote the submatrix of $A$ consisting of the first $j$ columns of $A$, and denote the $j^{th}$ column of $A$ by $a_j$.
Then, ${{\mathbf{q}}}$ is feasible if for all $j \in [n]$,
\begin{align*}
	A^j{{\mathbf{q}}}(\mathbf r^j) \leq b \qquad \forall \mathbf r^j {\in \times_{\ell \in [j]} R_\ell}.
\end{align*}
It will sometimes be helpful to write this expression down completely, i.e., for all $j \in [n]$, for all $k \in [K]$ and all profiles $\mathbf r^j {\in \times_{\ell \in [j]} R_\ell}$:
$$\sum_{\ell \in [j]}a_{k\ell}q_\ell(\mathbf r^\ell) \leq b_k.$$

Now, we describe a more compact formulation using the variables $Q_j(r_j) $ defined as follows:
$$Q_j(r_j) = \mathbb{E}_{\mathbf r^{j-1}} [q_j(\mathbf r^j)].$$
The variables $Q_j(r_j) $ are sometimes called {\bf interim allocation} variables. 

An interim allocation $\mathbf Q^n$ is {\bf implementable}  up until round $i \in [n]$, if there exists an ex-post allocation ${{\mathbf{q}}}$ satisfying the following set of constraints $(FP^i[\mathbf Q^i])$:
\begin{align*}
	&A^i{{\mathbf{q}}}(\mathbf r^i) \leq b & \forall \mathbf r^i {\in \times_{\ell \in [i]} R_\ell},\\
	&\mathbb{E}_{\mathbf r^{j-1}} [q_j(\mathbf r^j)] \geq Q_j(r_j) & \forall r_j {\in R_j}, \; \forall j \in [i],  \\
	&{{\mathbf{q}}}(\mathbf r^i) \geq 0 & \forall \mathbf r^i {\in \times_{\ell \in [i]} R_\ell}.
\end{align*}
If the interim allocation $\mathbf Q^n$ is implementable up until round $i$, we will say that $FP^i[\mathbf Q^i] \neq \emptyset$. An interim allocation $\mathbf Q^n$ is implementable if $FP^n[\mathbf Q^n] \neq \emptyset$.

Consider the following optimization problem:
\begin{align}
	h_{i+1}(\mathbf Q^i) = \max_{z, \mathbf q}\; &\mathbb{E}_{\mathbf r^i} [z(\mathbf r^i)] \nonumber\\
	\text{s.t.}\; &A^i{{\mathbf{q}}}(\mathbf r^i) + a_{i+1}z(\mathbf r^i) \leq b &\forall \mathbf r^i {\in \times_{\ell \in [i]} R_\ell},  \label{cap}\\
	&\mathbb{E}_{\mathbf r^{j-1}}[q_j(\mathbf r^j)] \geq Q_j(r_j) &\forall r_j {\in R_j}, \; \forall j \in [i], \label{expec}\\ 
	&{{\mathbf{q}}}, z \geq 0. \label{nonnegative}
\end{align}

\begin{lemma}\label{lemma:imp}
	$\mathbf Q^n$ is implementable if and only if $Q_{i+1}(r_{i+1}) \leq h_{i+1}(\mathbf Q^i) $ for all rewards $r_{i+1} {\in R_{i+1}}$ and all $i \in [n-1]$.
\end{lemma}
\begin{proof}
	If $\mathbf Q^n$ is implementable, the statement is clearly true. So, suppose to the contrary that $Q_{i+1}(r_{i+1}) \leq h_{i+1}(\mathbf Q^i)$ for all rewards $r_{i+1} {\in R_{i+1}}$ and all $i \in [n-1]$, but $\mathbf Q^n$ is not implementable, $i.e.$, $FP^{n}[\mathbf Q^n] = \emptyset$. For any $i \in [n]$, let $({{\mathbf{q}}},z)$ be an optimal solution that yields $h_{i+1}(\mathbf Q^i)$ and therefore satisfies (\ref{cap}-\ref{nonnegative}). Notice, ${{\mathbf{q}}}$ is a feasible ex-post allocation in $FP^i[\mathbf Q^i]$. Then, there does not exist any $q_{i+1}$ that solves the following system of inequalities together with ${{\mathbf{q}}}$:
	\begin{align}
		&A_i{{\mathbf{q}}}(\mathbf r^i) + a^{i+1}q_{i+1}(\mathbf r^{i+1}) \leq b &\forall \mathbf r^{i+1} {\in \times_{\ell \in [i+1]} R_\ell}, \label{upper}\\
		&\mathbb{E}_{\mathbf r^{i}}[q_{i+1}(\mathbf r^{i+1})] \geq Q_{i+1}(r_{i+1}) &\forall r_{i+1} {\in R_{i+1}}, \label{lower}\\
		&{{\mathbf{q}}}, z \geq 0. \label{zero}
	\end{align}
	
	Let $y_{\mathbf r^i}$ be the optimal dual variables associated with (\ref{cap}). Then, for all $\mathbf r^{i+1} {\in \times_{\ell \in [i+1]} R_\ell}$:
	\begin{align*}
		y_{\mathbf r^i}A^i{{\mathbf{q}}}(\mathbf r^i) + y_{\mathbf r^i}a_{i+1}q_{i+1}(\mathbf r^{i+1}) &\leq y_{\mathbf r^i}b,\\
		q_{i+1}(\mathbf r^{i+1}) &\leq \frac{y_{\mathbf r^i}b - y_{\mathbf r^i}A_i{\vec{\mathbf{q}}}(\mathbf r^i) }{y_{\mathbf r^i}a_{i+1}}.
	\end{align*}
    {
    Note that since $z$ has a positive objective coefficient and appears only in the constraints with coefficient $a_{i+1} \geq 0$, for $h_{i+1}(\mathbf Q^i)$ to be bounded, there must exist an optimal dual solution with $y_{\mathbf r^i}a_{i+1} >0$. Hence, we have a well-defined
    }
    upper bound for each $q_{i+1}(\mathbf r^{i+1})$. Knowing that (\ref{upper}-\ref{zero}) is infeasible implies that even if we set each $q_{i+1}(\mathbf r^{i+1})$ at its upper bound, we must violate one of the constraints in (\ref{lower}). Hence, there is an $r_{i+1}$ such that:
	\begin{align*}
		Q_{i+1}(r_{i+1}) > \mathbb{E}_{\mathbf r^{i}}[ \frac{y_{\mathbf r^i}b - y_{\mathbf r^i}A_i{{\mathbf{q}}}(\mathbf r^i)}{y_{\mathbf r^i}a_{i+1}}]
		= \mathbb{E}_{\mathbf r^{i}}[z(\mathbf r^i)] = h_{i+1}(\mathbf Q^i).
	\end{align*}
	The penultimate equation follows by complementary slackness, and we obtain a contradiction.
\end{proof}

Now, we describe a linear programming formulation for the {\em off-line} selection problem. In the off-line version, the choice of the $i^{th}$ service level can be based on the {\em entire} profile of rewards. Denote by $w_{i}(\mathbf r^n)$, the ex-post service allocated to $i$ at reward profile $\mathbf r^n \in \times_{\ell \in [n]} R_\ell$. Let ${{\mathbf{w}}}(\mathbf r^n) = [w_1(\mathbf r^ n), w_2(\mathbf r^ n), \ldots, w_n(\mathbf r^n)].$
As before, we define the interim allocation as follows:
$$
W_i(r_i) = \mathbb{E}_{\mathbf r_{-i}}[w_i(r_i, \mathbf r_{-i})].
$$
The off-line selection problem (often called the prophet's problem) can be expressed as follows:
\begin{align*}
	\max_{\mathbf w, \mathbf W} \; &\sum_{j \in [n]} \mathbb{E}_{r_j}[r_jW_j(r_j)]\\
	\text{s.t.}\; &A^i{{\mathbf{w}}}(\mathbf r^n) \leq b & \forall \mathbf r^n {\in \times_{\ell \in [n]} R_\ell}, \; \forall i \in [n],\\
	&W_i(r_i) = \mathbb{E}_{\mathbf r_{-i}}[w_i(r_i, \mathbf r_{-i})] &\forall r_i {\in R_i}, \; \forall i \in [n],\\
	&{{\mathbf{w}}} \geq 0.
\end{align*}
Denote the optimal solution to the off-line selection problem by $({{\mathbf{w}}}^*, {\mathbf{W}}^*)$. 

If we prove that $h_{i+1}(0.5 {\mathbf{W}^i}^*) \geq 0.5W^*_{i+1}(r_{i+1})$ for all $r_{i+1} {\in R_{i+1}}$, then, from Lemma \ref{lemma:imp}, it follows that an on-line solution with a value of at least half the optimal off-line solution is implementable. In other words, we obtain a prophet inequality with a bound of $1/2$. To illustrate, we apply the idea to the case where $K=1$, $i.e.$, there is only one constraint.\footnote{\citet{epitropou2019optimal} considered this case with the additional restriction that $a_{1i}=1$ for all $i$.} To see a variation of the following result and its application to other settings, the interested reader can refer to \citet{epitropou2020essays}.

\begin{prop}\label{prop:oneConstrained}
    If $h_{i+1}(0.5 {\mathbf{W}^i}^*) \geq 0.5W^*_{i+1}(r_{i+1})$ for all $r_{i+1} {\in R_{i+1}}$ and all $i \in [n-1]$, then, the interim allocation $0.5 {\mathbf{W}}^*$ is implementable.
\end{prop}
\begin{proof}
	Notice that $0.5 {\mathbf{W}^2}^*$ is feasible in any on-line problem with nonnegative $A$ thanks to scaling by $0.5$. Hence, we can proceed as follows. Assume that there exists some $i<n$ such that $h_{j}(0.5 {\mathbf{W}^{j-1}}^*) \geq 0.5W^*_{j}(r_{j})$ for all $r_{j} \in R_j$ and $j \in [i]$. Then, the problem $h_{i+1}(0.5 {\mathbf{W}^i}^*)$ is feasible. Therefore, for any feasible ${\mathbf q}$ with the expected values $0.5 {\mathbf{W}^i}^*$, it is optimal to set:
    \begin{align*}
        z^*(\mathbf r^i) &= \big(b- \sum_{j \in [i]} a_{j} q_j(\mathbf r^j) \big)/a_{i+1} & \forall \mathbf r^i {\in \times_{\ell \in [i]} R_\ell},
    \end{align*}
	as there is only one constraint for each $\mathbf r^i$. Then, the objective becomes:
	\begin{align*}
		h_{i+1}(0.5 {\mathbf{W}^i}^*) &= \mathbb{E}_{\mathbf r^i} [z^*(\mathbf r^i)] = \mathbb{E}_{\mathbf r^i} [b- \sum_{j \in [i]} a_{j} q_j(\mathbf r^j)]/a_{i+1}\\
		&=\big(b- \sum_{j \in [i]} a_{j} \mathbb{E}_{r_j} [0.5 W^*_j(r_j)] \big)/a_{i+1} \\
        &= 0.5 b/a_{i+1} \geq 0.5 W^*_{i+1}(r_{i+1}) &\forall r_{i+1} {\in R_{i+1}}.
	\end{align*}
	Hence, via induction and Lemma \ref{lemma:imp}, we conclude that $0.5 {\mathbf{W}}^*$ is implementable.
\end{proof}

The linear program for the stopping problem used in \citet{dk} begins with the ex-post variables for the off-line problem, $w_{i}(\mathbf r^n)$, and imposes constraints to convert them to into ex-post variables for the on-line version. Specifically, for each $i \in [n]$ and $(\mathbf r^i)$, they require
$$w_i(r_1, \ldots, r_i, t_{i+1}, \ldots, t_n) = w_i(r_1, \ldots, r_i, s_{i+1}, \ldots, s_n)$$
for all $(t_{i+1}, \ldots, t_n)$ and $(s_{i+1}, \ldots, s_n)$. These are called non-anticipativity constraints. In our case, we work with the interim allocation variables. The off-line and on-line versions share the same number of interim allocation variables but differ in their underlying feasible region. We establish a prophet inequality by showing that the interim allocation variables for the off-line problem can be scaled so as to be feasible for the on-line problem.

\subsection{Minkowski Sums}

The Minkowski sum of polytopes $P$ and $Q$ in $\mathbb{R}^n$ is defined as $\alpha P+ \beta Q = \{\alpha \mathbf p+ \beta \mathbf q: \; \mathbf p \in P,\; \mathbf q \in Q\}$ for any $\alpha, \beta \in \mathbb{R}$. 
In this section, we show how Minkowski sums preserve results on prophet inequalities. We first introduce the concept. 

\begin{thrm}\label{theorem: Minkowski_Result}
	Consider an on-line selection problem whose underlying polyhedron is the Minkowski sum of polytopes $\sum_{m \in [M]} \alpha_m P_m \subseteq \mathbb{R}^n$ for some $M \in \mathbb{N}$. Let $\alpha_m \geq 0$, and suppose that the on-line selection problem defined on the polytope $P_m$ admits a prophet inequality with factor of $p_m \in [0,1]$ for all $m \in [M]$. Then, the on-line selection problem defined over the Minkowski sum, $\sum_{m \in [M]} \alpha_m P_m$, admits a prophet inequality with factor at least $\min_{m \in [M]} \{p_m\}$.
\end{thrm}
\begin{proof}
	Let $\mathbf w^* \in \sum_{m \in [M]} \alpha_m P_m$ be the optimal off-line solution. The Minkowski sum definition implies that there exist $\mathbf w_m^* \in P_m$ for all $m \in [M]$ such that $\mathbf w^*= \sum_{m \in [M]} \alpha_m \mathbf w_m^*$. 
    
    As the Minkowski sum coefficients, $\{\alpha_m\}_{m \in [M]}$, are all non-negative, and the objective function of the problem is linear over $\mathbf w^*$, the optimal off-line solution of the problem when the prophet is limited to $P_m$ should coincide with $\mathbf w_m^*$ for all $m \in [M]$. If we let $Z_{offline}(P)$ denote the optimal objective value of the off-line problem under some polyhedron $P$, then the following equality must hold:
    \begin{align*}
        Z_{offline}\Big(\sum_{m \in [M]} \alpha_m P_m\Big) = \sum_{m \in [M]} \alpha_m Z_{offline}(P_m).
    \end{align*}
	Letting $Z_{online}(P)$ denote the optimal objective value of the on-line problem under some polyhedron $P$, we also have:
    \begin{align*}
        p_m Z_{offline}(P_m) &\leq Z_{online}(P_m) &\forall m \in [M].
    \end{align*}
	Therefore, we can conclude our proof using the given prophet inequality factors:
    \begin{align*}
        Z_{online}\Big( \sum_{m \in [M]} \alpha_m P_m \Big) &= \sum_{m \in [M]} \alpha_m Z_{online}(P_m) \geq \sum_{m \in [M]} \alpha_m p_m Z_{offline}(P_m)\\
        &\geq \min_{m \in [M]} \{p_m\} Z_{offline}\Big(\sum_{m \in [M]} \alpha_m P_m\Big),
    \end{align*}
    where the first equality follows from the fact that the Minkowski sum of the optimal on-line solutions, $\mathbf q_m^* \in P_m$ for $m \in [M]$, is an optimal on-line solution of $\sum_{m \in M} \alpha_m P_m$.
\end{proof}

{Theorem \ref{theorem: Minkowski_Result} provides a simple way to derive prophet inequalities for certain polytopes. Here is an example. The Minkowski sum of simplices is a permutahedron, see \citet{AguiarArdila2017}. As the on-line selection problem defined over a simplex admits a factor $\frac{1}{2}$ prophet inequality, it follows by Theorem \ref{theorem: Minkowski_Result}, that the on-line selection problem defined over a permutahedron, which admits nonnegative Minkowski sum coefficients, also admits a factor $\frac{1}{2}$ prophet inequality. Also, the factor $\frac{1}{2}$ prophet inequality for uniform matroids is an immediate consequence (see \citet{alaei2014bayesian} for a recent treatment).}

\section{Applications}
We demonstrate the power of the approach by showing how to derive prophet inequalities for a variety of settings. 

\subsection{$K$ Many Constraints}

In this section, we prove that the selection problem with independently drawn rewards and $K$ many constraints admits a $\frac{1}{K+1}$-prophet inequality. By Proposition \ref{prop:oneConstrained}, it suffices to show that $h_{i+1}(\frac{1}{K+1} {\mathbf{W}^i}^*) \geq \frac{1}{K+1} W^*_{i+1}(r_{i+1})$ for all $r_{i+1} \in R_{i+1}$.

\begin{thrm}\label{thm:multiknapsack}
	Consider a selection problem with independently drawn nonnegative rewards where the feasible on-line policies live in a polyhedron described by $A \in \mathbb{R}_+^{K \times n}$ and $b \in \mathbb{R}^n$. The optimal on-line strategy achieves at least $\frac{1}{K+1}$ of the prophet's value.
\end{thrm}
\begin{proof}
First, notice that $\frac{1}{K+1} {\mathbf{W}^{K+1}}^*$ is feasible thanks to scaling by $\frac{1}{K+1}$. If $n \leq K+1$, we are done. Otherwise, as an induction hypothesis, we can assume that there exists some $i<n$ such that:
$$
h_{j}(\frac{1}{K+1} {\mathbf{W}^{j-1}}^*) \geq \frac{1}{K+1}W^*_{j}(r_{j}) \qquad \forall r_j \in R_j, \; \forall j \in [i].
$$
Then, the problem $h_{i+1}(\frac{1}{K+1}{\mathbf{W}^{i}}^*)$ is feasible, and for any feasible ${\mathbf q}$ with the interim values $\frac{1}{K+1} {\mathbf{W}^{i}}^*$, it is optimal to set:
$$
z^*(\mathbf r^i) = \min_k \{b_k- \sum_{j \in [i]} a_{kj} q_j(\mathbf r^j)\} \qquad \forall \mathbf r^i \in \times_{\ell \in [i]} R_\ell.
$$
Note that, to simplify the notation, we normalize the problem parameters so that the coefficient of $z$ becomes $a_{i+1} =1$. This is without loss.

We partition the set of all reward profiles $\mathcal{R}^i= \times_{j \in [i]} R_j$ into $K$ sets with respect to the tightest bound on $z^*$ variables. {Any partition, $\{\mathcal{R}_k\}_{k \in [K]}$, that satisfy the following property would suffice:
\begin{align*}    
    \mathcal{R}_k &\subseteq  \{\mathbf r^i \in \mathcal{R}^i : \; z^*(\mathbf r^i)= b_k- \sum_{j \in [i]} a_{kj} q_j(\mathbf r^j) \} & \forall k \in [K].
\end{align*}
}
Finally, we define probability $f_k = \sum_{\mathbf r^i \in \mathcal{R}_k} f_{\mathbf r^i}$ for all $k \in [K]$ and set $\mathcal{K}=\{k \in [K] \; : \; f_k \geq \frac{1}{K+1}\}$. Notice that $\mathcal{K}$ cannot be empty since otherwise we would have the following contradiction:
$$
\sum_{k \in [K]} f_k = \sum_{k \in [K]} \sum_{\mathbf r^i \in \mathcal{R}_k} f_{\mathbf r^i} = \sum_{\mathbf r^i \in \mathcal{R}^i} f_{\mathbf r^i} < \sum_{k \in [K]} \frac{1}{K+1} = \frac{K}{K+1}< 1.
$$
Then, we can use $\{\mathcal{R}_k\}_{k \in [K]}$ to derive a lower bound for the optimal objective function:
\begin{align*}
	h_{i+1}(\frac{1}{K+1} {\mathbf{W}^i}^*) &= \mathbb{E}_{\mathbf r^i} [z^*(\mathbf r^i)] = \sum_{k \in [K]} \sum_{\mathbf r^i \in \mathcal{R}^i} f_{\mathbf r^i} [b_k- \sum_{j \in [i]} a_{kj} q_j(\mathbf r^j)]\\
	&= \sum_{k \in [K]} f_k b_k - \sum_{k \in [K]} \sum_{\mathbf r^i \in \mathcal{R}_k} f_{\mathbf r^i} \sum_{j \in [i]} a_{kj} q_j(\mathbf r^j),\\
    &\geq \sum_{k \in \mathcal{K}} f_k b_k - \sum_{k \in \mathcal{K}} \sum_{\mathbf r^i \in \mathcal{R}_k} f_{\mathbf r^i} \sum_{j \in [i]} a_{kj} q_j(\mathbf r^j),\\
	&\geq \sum_{k \in \mathcal{K}} f_k b_k - \sum_{k \in \mathcal{K}} \sum_{\mathbf r^i \in \mathcal{R}^i} f_{\mathbf r^i} \sum_{j \in [i]} a_{kj} q_j(\mathbf r^j),\\
	&\geq \sum_{k \in \mathcal{K}} f_k b_k - \sum_{k \in \mathcal{K}} \sum_{j \in [i]} a_{kj} \sum_{r_j \in R_j} f_{r_j} \frac{W_j(r_j)}{K+1}.
\end{align*}
We know from the off-line problem that $\sum_{j \in [i]} a_{kj} \sum_{r_j} f_{r_j} \frac{W_j(r_j)}{K+1} \leq \frac{b_k}{K+1}$ for all $k \in [K]$. Hence, the following lower bound follows:
\begin{align*}
	h_{i+1}(\frac{1}{K+1} {\mathbf{W}^i}^*) &\geq \sum_{k \in \mathcal{K}} \Big(f_k-\frac{1}{K+1}\Big) b_k \\
    &\geq \min_{\ell \in [K]} b_\ell \sum_{k \in \mathcal{K}} \Big(f_k-\frac{1}{K+1}\Big),\\
	&= \min_{\ell \in [K]} b_\ell \Big(\sum_{k \in \mathcal{K}} f_k-\frac{|\mathcal{K}|}{K+1}\Big)  \\
	&= \min_{\ell \in [K]} b_\ell \Big(\frac{K+1-|\mathcal{K}|}{K+1} - \sum_{k \in [K] \setminus \mathcal{K}} f_k\Big) \\
    &\geq \frac{1}{K+1} \min_{\ell \in [K]} b_\ell \geq \frac{1}{K+1}W^*_{i+1}(r_{i+1}) & \forall r_{i+1} \in R_{i+1},
\end{align*}
where the penultimate inequality follows from the fact that $\sum_{k \in [K] \setminus \mathcal{K}} f_k \leq \sum_{k \in [K] \setminus \mathcal{K}} \frac{1}{K+1} =\frac{K-|\mathcal{K}|}{K+1}$. Hence, by induction and Lemma \ref{lemma:imp}, we conclude that $\frac{\mathbf W^*}{K+1}$ is implementable.
\end{proof}

\subsection{Polymatroids} \label{Subsection:PolymatroidsProphet}

Suppose that the underlying polyhedron of the selection problem is a polymatroid defined by a non-decreasing, non-negative submodular function,  $g: 2^{n} \mapsto \mathbb{R}$ where $2^n$ denotes the power set of $[n]$. 
In this case, we recover the $\frac{1}{2}$ prophet inequality of \citet{dutting2015polymatroid} as a Corollary of a subsequent result, Theorem \ref{thrm: onlinePolymatroidMainResult}. In fact our result is slightly stronger in that we do not assume that the submodular function is integer valued. 
\begin{Corollary}
    If the underlying polyhedron of the selection problem is defined by a polymatroid, then the optimal on-line strategy achieves at least $\frac{1}{2}$ of the prophet's value.
\end{Corollary}

The more general result allows the underlying polyhedron itself to depend upon the realized rewards.

\subsection{On-line Polymatroids}

We introduce a new setting where the underlying polymatroid {\em depends} on the realized rewards. Formally, 
let $g: 2^{[n]} \times (\times_{\ell \in [n]} R_\ell)  \mapsto \mathbb{R}$ be a non-negative, non-decreasing, submodular function for any given reward profile. That is, $g$ satisfies the following properties for all $\mathbf r^n \in \times_{\ell \in [n]} R_\ell$:
\begin{itemize}
    \item $g(I,\mathbf r^n) \geq 0$ for all $I \subseteq [n]$,
    \item $g(I,\mathbf r^n) \leq g(J,\mathbf r^n)$ for all $I \subseteq J \subseteq [n]$,
    \item $g(I,\mathbf r^n) + g(J,\mathbf r^n) \geq g({I \cup J},\mathbf r^n) + g({I \cap J},\mathbf r^n)$ for all $I, J \subseteq [n]$.
\end{itemize}
Hence, $g$ can be defined through different submodular functions defined for every reward profile. 

To ensure that $g$ depends only on the realized rewards, we suppose that for all $\mathbf r^n \in \times_{\ell \in [n]} R_\ell$ and $I \subseteq [n]$:
\begin{align*}
    &g(I,\mathbf r^{I},\mathbf r^{-I}) = g(I,\mathbf r^{I},\hat{\mathbf r}^{-I}) & \forall \hat{\mathbf r}^{-I} \in \times_{\ell \in [n] \setminus I} R_\ell,
\end{align*}
where $\mathbf r^{I} \in \times_{\ell \in I} R_\ell$ and $\mathbf r^{-I} \in \times_{\ell \in [n] \setminus I} R_\ell$. Call such a $g$ an \textit{on-line submodular function}, and the polymatroid defined by it an \textit{on-line polymatroid}. We will denote any such $g$ through a simplified notation, $g(I,\mathbf r^{I})$, as the remaining profile of rewards does not affect the value of $g$. We consider the case when the selection is constrained to lie in an on-line polymatroid. Note that a static polymatroid, which does not depend on rewards, is a special case of an on-line polymatroid. Hence, our results in this section also apply to them.

\begin{example}
    The following are some examples of on-line submodular functions that are nonnegative and nondecreasing. Note that we assume all rewards to be nonnegative.
    \begin{equation*}
        \begin{aligned}
        &g_1(I,\mathbf{r}^I) = \min\{1 + \sum_{\ell \in I} r_\ell, B\} &\forall \mathbf r^I \in \times_{\ell \in I} R_\ell, \forall I \subseteq [n],\\
        &g_2(I,\mathbf{r}^I) = \begin{cases}
            n &\text{if } \max_{\ell \in I} r_\ell \geq T,\\
            |I| &\text{if } \max_{\ell \in I} r_\ell < T,\\
        \end{cases} &\forall \mathbf r^I \in \times_{\ell \in I} R_\ell, \forall I \subseteq [n],\\
        &g_3(I,\mathbf{r}^I) = \begin{cases}
            \sum_{\ell \in I} \mathds{1}_{\{r_\ell \geq T\}} &\text{if } \min_{\ell \in I} r_\ell \geq T,\\
            n &\text{if } \min_{\ell \in I} r_\ell < T,\\
        \end{cases} &\forall \mathbf r^I \in \times_{\ell \in I} R_\ell, \forall I \subseteq [n],\\
        \end{aligned}
    \end{equation*}
    where $B$ and $T$ are fixed real numbers.
\end{example}

On-line polymatroids allow one to capture settings where the available resources for servicing requests can change with the realized rewards. For example, bringing in extra servers or ramping up peaker plants.

We show that the approach presented in Section \ref{section: model} extends to this setting by proving that $h_{i+1}(0.5 {\mathbf{W}^i}^*) \geq 0.5g(\{i+1\}, r_{i+1}) \geq 0.5W^*_{i+1}(r_{i+1})$ for all $r_{i+1} {\in R_{i+1}}$ and all $i \in [n-1]$ (see Proposition \ref{prop:oneConstrained}).
This yields a $\frac{1}{2}$ prophet inequality, which, as far as we know, is new.

We fix an arbitrary $\hat r_{i+1} {\in R_{i+1}}$ and consider the following model, which we denote interchangeably by $h_{i+1}(0.5{\mathbf{W}^i}^*)$ and \eqref{eq: MainProblem}:
\begin{align}
\max_{z, \mathbf{q}}\; &\mathbb{E}_{\mathbf r^i} [z(\mathbf r^i)] \label{eq: MainProblem} \tag{$\mathcal{P}$} \\
	\text{s.t.}\; &\sum_{j \in I} q_j(\mathbf r^j) + z(\mathbf r^i) \leq g(I \cup \{i+1\},\mathbf r^I) &\forall \mathbf r^i {\in \times_{\ell \in [i]} R_\ell}, \forall I \subseteq[i], \nonumber \\
    &\sum_{j \in I} q_j(\mathbf r^j) \leq g(I,\mathbf r^I) &\forall \mathbf r^i {\in \times_{\ell \in [i]} R_\ell}, \forall I \subseteq[i], \nonumber \\
	&\mathbb{E}_{\mathbf r^{j-1}}[q_j(\mathbf r^j)] = 0.5W_j^*(r_j) &\forall r_j {\in R_j}, \; \forall j \in [i], \nonumber \\ 
	&{{\mathbf{q}}}, z \geq 0. \nonumber
\end{align}
In this formulation we have suppressed the dependence of $g$ on $\hat{r}_{i+1}$. This is without loss because, as we are interested in proving $h_{i+1}(0.5 {\mathbf{W}^i}^*) \geq 0.5g(\{i+1\}, \hat r_{i+1})$, we can focus only on the reward profiles containing $\hat r_{i+1}$. Hence, we can simplify the notation of $g(I \cup \{i+1\},\mathbf r^I, \hat{r}_{i+1})$ by denoting it as $g(I \cup \{i+1\},\mathbf r^I)$. This is only a notational change to align with Section \ref{section: model}, which has no consequence on the rest of our arguments. Accordingly, we also drop the dependence of the expected value of $g$ on $\hat r_{i+1}$, which is denoted as follows:
\begin{align*}
    &g(I \cup \{i+1\}) = \mathbb{E}_{\mathbf r^I}[g(I \cup \{i+1\}, \mathbf r^{I})] &\forall I \subseteq [i].
\end{align*}

We have also restricted the expected allocation constraints to equality. This restriction does not change the value of $h_{i+1}(0.5{\mathbf{W}^i}^*)$, as the objective function is non-increasing in $\mathbf{q}$ variables. Next, we will define a series of other formulations and prove their relation to \eqref{eq: MainProblem} via a series of Lemmas. Eventually, we will end up with a model enabling our result via induction and simple feasibility arguments. 

We now define another model, \eqref{eq: EquivalentProblem}, that partitions $\mathbf{q}$ into two parts and further restricts one of them through a new constraint \eqref{eq: q''constraint}. Our next result will prove that solving \eqref{eq: EquivalentProblem} is equivalent to solving \eqref{eq: MainProblem}:
\begin{align}
    \max_{\mathbf{z}, \mathbf{q}', \mathbf{q}''}\; &\mathbb{E}_{\mathbf r^i} [z(\mathbf r^i)] \label{eq: EquivalentProblem} \tag{$\mathcal{P}'$}\\
	\text{s.t.}\; &\sum_{j \in I} \Big[q_j'(\mathbf r^j) + q_j''(\mathbf r^j) \Big] + z(\mathbf r^i) \leq g(I \cup \{i+1\},\mathbf r^I) &\forall \mathbf r^i {\in \times_{\ell \in [i]} R_\ell}, \forall I \subseteq[i], \nonumber\\
    &\sum_{j \in I} q_j''(\mathbf r^j) \leq g(I \cup \{i+1\},\mathbf r^I) - g(\{i+1\}) &\forall \mathbf r^i {\in \times_{\ell \in [i]} R_\ell}, \forall I \subseteq[i], \label{eq: q''constraint}\\
    &\sum_{j \in I} \Big[q_j'(\mathbf r^j) + q_j''(\mathbf r^j) \Big] \leq g(I,\mathbf r^I) &\forall \mathbf r^i {\in \times_{\ell \in [i]} R_\ell}, \forall I \subseteq[i],  \nonumber \\
	&\mathbb{E}_{\mathbf r^{j-1}} \Big[q_j'(\mathbf r^j) + q_j''(\mathbf r^j) \Big] = 0.5{W_j^*(r_j)} &\forall r_j {\in R_j}, \; \forall j \in [i], \nonumber \\ 
	&{{\mathbf{q}}}', {{\mathbf{q}}}'', z \geq 0. \nonumber
\end{align}

\begin{lemma}\label{lem: EquivalenceResult}
    If $g(I,\mathbf r^i)$ is a nonnegative and nondecreasing set function for all $\mathbf r^i {\in \times_{\ell \in [i]} R_\ell}$, and \eqref{eq: MainProblem} is feasible, then the formulations \eqref{eq: MainProblem} and \eqref{eq: EquivalentProblem} have the same optimal objective value.
\end{lemma}
\begin{proof}
    As \eqref{eq: MainProblem} and \eqref{eq: EquivalentProblem} have the same objective function over $z$, we prove the result by showing that their feasibility sets are essentially the same.

    Let $(\mathbf{z}, \mathbf{q})$ be a feasible solution of $\eqref{eq: MainProblem}$. Then, it is easy to see that $(\mathbf{z}, \mathbf{q}' = \mathbf{q}, \mathbf{q}''=0)$ is feasible in \eqref{eq: EquivalentProblem} and yields the same objective value as $(\mathbf{z}, \mathbf{q})$. This is because $g$ is nonnegative and monotone increasing for all $\mathbf r^i {\in \times_{\ell \in [i]} R_\ell}$, which means that the right-hand side of \eqref{eq: q''constraint} is nonnegative: $g(I \cup \{i+1\},\mathbf r^I) - g(\{i+1\}) \geq 0$ for all $I \subseteq[i]$ and for all $\mathbf r^i {\in \times_{\ell \in [i]} R_\ell}$.

    Now, let $(\mathbf{z}, \mathbf{q}', \mathbf{q}'')$ be a feasible solution of $\eqref{eq: EquivalentProblem}$. Then, $(\mathbf{z},\mathbf{q} = \mathbf{q}' + \mathbf{q}'')$ is feasible in \eqref{eq: MainProblem} and yields the same objective value as $(\mathbf{z}, \mathbf{q}', \mathbf{q}'')$. 

    Recall that we assume \eqref{eq: MainProblem} to be feasible, and it is bounded by definition. Then, as any solution of one problem (\eqref{eq: MainProblem} or \eqref{eq: EquivalentProblem}) can be recast as a solution of the other with the same objective value, they should have the same optimal objective value.
\end{proof}

Thanks to Lemma \ref{lem: EquivalenceResult}, we can work with \eqref{eq: EquivalentProblem} instead of \eqref{eq: MainProblem}. Next, we present another model that yields a lower bound for \eqref{eq: EquivalentProblem}:
\begin{align}
    \max_{\mathbf{q}', \mathbf{q}''}\; &\mathbb{E}_{\mathbf r^i} [g(\{i+1\}) - \sum_{j \in [i]} q_j'(\mathbf r^j)] \label{eq: LowerBound} \tag{$\mathcal{P}_\ell'$}\\
	\text{s.t.}\; &\sum_{j \in I} \Big[q_j'(\mathbf r^j) + q_j''(\mathbf r^j) \Big] \leq g(I \cup \{i+1\},\mathbf r^I) &\forall \mathbf r^i {\in \times_{\ell \in [i]} R_\ell}, \forall I \subseteq[i], \nonumber \\
    &\sum_{j \in I} q_j''(\mathbf r^j) \leq g(I \cup \{i+1\},\mathbf r^I) - g(\{i+1\}) &\forall \mathbf r^i {\in \times_{\ell \in [i]} R_\ell}, \forall I \subseteq[i], \tag{\ref{eq: q''constraint}}\\
    &\sum_{j \in I} \Big[q_j'(\mathbf r^j) + q_j''(\mathbf r^j) \Big] \leq g(I,\mathbf r^I) &\forall \mathbf r^i {\in \times_{\ell \in [i]} R_\ell}, \forall I \subseteq[i],  \nonumber \\
	&\mathbb{E}_{\mathbf r^{j-1}} \Big[q_j'(\mathbf r^j) + q_j''(\mathbf r^j) \Big] = 0.5{W_j^*(r_j)} &\forall r_j {\in R_j}, \; \forall j \in [i], \nonumber \\ 
	&{{\mathbf{q}}}', {{\mathbf{q}}}'' \geq 0. \nonumber
\end{align}

\begin{lemma}\label{lem: LowerBoundModel}
    The optimal solution of \eqref{eq: LowerBound} lower bounds the optimal solution of \eqref{eq: EquivalentProblem}.
\end{lemma}
\begin{proof}
    If \eqref{eq: LowerBound} is infeasible, then \eqref{eq: EquivalentProblem} also does not have any feasible solution. In this case, we can say that their optimal solutions both yield $- \infty$, so the result follows. Hence, below, we treat the case of feasible \eqref{eq: LowerBound}.

    Any feasible solution, $(\mathbf{q}',\mathbf{q}'')$, of \eqref{eq: LowerBound}, can be coupled with some $\mathbf{z}$ to find a feasible solution of \eqref{eq: EquivalentProblem}. To maximize the objective, one can define:
    \begin{align}\label{eq: ConstructedZ}
        &z(\mathbf r^i) = \min_{I \subseteq [i]} \Big\{ g(I \cup \{i+1\},\mathbf r^I) - \sum_{j \in I} \big[q_j'(\mathbf r^j) + q_j''(\mathbf r^j) \big] \Big\} &\forall \mathbf r^i {\in \times_{\ell \in [i]} R_\ell}.
    \end{align}

    We now show that, given any feasible $(\mathbf{q}',\mathbf{q}'')$ in \eqref{eq: LowerBound}, the feasible solution of \eqref{eq: EquivalentProblem} constructed as in \eqref{eq: ConstructedZ} yields a weakly higher objective value. Hence, \eqref{eq: EquivalentProblem} should have an optimal objective value that is lower bounded by that of \eqref{eq: LowerBound}. To this end, consider the objective of \eqref{eq: EquivalentProblem} given by the equations in \eqref{eq: ConstructedZ} for any fixed $(\mathbf{q}',\mathbf{q}'')$:
    \begin{align*}
        \mathbb{E}_{\mathbf r^i}[z(\mathbf r^i)] &= \mathbb{E}_{\mathbf r^i}\Bigg[ \min_{I \subseteq [i]} \Big\{ g(I \cup \{i+1\},\mathbf r^I) - \sum_{j \in I} \big[q_j'(\mathbf r^j) + q_j''(\mathbf r^j) \big] \Big\} \Bigg] \\
        &\geq \mathbb{E}_{\mathbf r^i}\Bigg[ \min_{I \subseteq [i]} \Big\{ g(I \cup \{i+1\},\mathbf r^I) - \sum_{j \in I} q_j'(\mathbf r^j) \\
        & \qquad \qquad \qquad \qquad \qquad \qquad \qquad -\big[ g(I \cup \{i+1\},\mathbf r^I) - g(\{i+1\}) \big] \Big\} \Bigg]\\
        &\geq \mathbb{E}_{\mathbf r^i}\Bigg[ \min_{I \subseteq [i]} \Big\{ g( \{i+1\}) - \sum_{j \in I} q_j'(\mathbf r^j) \Big\} \Bigg] \geq  \mathbb{E}_{\mathbf r^i}\Bigg[ g( \{i+1\}) - \sum_{j \in [i]} q_j'(\mathbf r^j) \Bigg], \\
    \end{align*}
    where the first inequality follows from \eqref{eq: q''constraint} for any $(\mathbf{q}',\mathbf{q}'')$ feasible in \eqref{eq: LowerBound}, and the last inequality follows as feasible $\mathbf{q}'$ are nonnegative. Hence, given any feasible solution of \eqref{eq: LowerBound}, one can find a feasible solution of \eqref{eq: EquivalentProblem} that yields a weakly higher objective value.
\end{proof}

We now consider a relaxation of \eqref{eq: LowerBound} and prove that it has the same optimal solution as \eqref{eq: LowerBound}:
\begin{align}
    \max_{\mathbf{q}', \mathbf{q}''}\; &\mathbb{E}_{\mathbf r^i} [g(\{i+1\}) - \sum_{j \in [i]} q_j'(r^j)] \label{eq: Relaxation} \tag{$\mathcal{P}_r'$}\\
	\text{s.t.}\; &\sum_{j \in I} \Big[q_j'(\mathbf r^j) + q_j''(\mathbf r^j) \Big] \leq g(I \cup \{i+1\},\mathbf r^I) &\forall \mathbf r^i {\in \times_{\ell \in [i]} R_\ell}, \forall I \subseteq[i], \nonumber \\
    &\sum_{j \in I} q_j''(\mathbf r^j) \leq g(I \cup \{i+1\},\mathbf r^I) - g(\{i+1\}) &\forall \mathbf r^i {\in \times_{\ell \in [i]} R_\ell}, \forall I \subseteq[i], \tag{\ref{eq: q''constraint}}\\
	&\mathbb{E}_{\mathbf r^{j-1}} \Big[q_j'(\mathbf r^j) + q_j''(\mathbf r^j) \Big] = 0.5{W_j^*(r_j)} &\forall r_j {\in R_j}, \; \forall j \in [i], \nonumber \\ 
	&{{\mathbf{q}}}', {{\mathbf{q}}}'' \geq 0. \nonumber
\end{align}

\begin{lemma}\label{lem: equivalentLowerBound}
    If \eqref{eq: MainProblem} is feasible, then \eqref{eq: LowerBound} and \eqref{eq: Relaxation} have the same optimal solution.
\end{lemma}
\begin{proof}
    Assume to the contrary that \eqref{eq: MainProblem} is feasible, but \eqref{eq: Relaxation} has an optimal solution that is strictly higher than that of \eqref{eq: LowerBound}. Let $(\mathbf{q}'^*,\mathbf{q}''^*)$ denote the optimal solution of \eqref{eq: Relaxation}, which cannot be feasible in \eqref{eq: LowerBound}. Then the following system of equations, where we utilize the values of $\mathbf{q}''^*$, should be infeasible:
    \begin{align*}
    &\sum_{j \in I} q_j'(\mathbf r^j)  \leq g(I \cup \{i+1\},\mathbf r^I) - \sum_{j \in I} q_j''^*(\mathbf r^j) &\forall \mathbf r^i {\in \times_{\ell \in [i]} R_\ell}, \forall I \subseteq[i], \\
    &\sum_{j \in I} q_j'(\mathbf r^j) \leq g(I,\mathbf r^I) - \sum_{j \in I} q_j''^*(\mathbf r^j) &\forall \mathbf r^i {\in \times_{\ell \in [i]} R_\ell}, \forall I \subseteq[i],  \nonumber \\
	&\mathbb{E}_{\mathbf r^{j-1}} \Big[q_j'(\mathbf r^j)\Big] = 0.5{W_j^*(r_j)} - \mathbb{E}_{\mathbf r^{j-1}} \Big[ q_j''^*(\mathbf r^j) \Big] &\forall r_j {\in R_j}, \; \forall j \in [i], \nonumber \\ 
	&{{\mathbf{q}}}' \geq 0. \nonumber
    \end{align*}
Then, via Farkas' Lemma, we know that the following system is feasible:
\begin{align}
    &\sum_{I \subseteq[i]} \sum_{\mathbf r^i {\in \times_{\ell \in [i]} R_\ell}} \Bigg[ \delta_I(\mathbf r^i) \Big( g(I \cup \{i+1\},\mathbf r^I) - \sum_{j \in I} q_j''^*(\mathbf r^j) \Big) \label{eq: firstFarkas'alternative} \tag{$\mathcal{F}$}  \\
    &\hspace{3.2cm}+ \sigma_I(\mathbf r^i) \Big( g(I,\mathbf r^I) - \sum_{j \in I} q_j''^*(\mathbf r^j) \Big) \Bigg] \nonumber \\
    &\hspace{3.6cm} + \sum_{j \in [i]} \sum_{r_j \in R_j} \lambda_j(r_j) \Bigg[ 0.5{W_j^*(r_j)} - \mathbb{E}_{\mathbf r^{j-1}}[q_j''^*(\mathbf r^j)]\Bigg]<0, \label{eq:  FarkasConstraintStrict} \\
    &\sum_{I \ni j} \sum_{(r_{j+1},\ldots,r_i)} \big[ \delta_I(\mathbf r^i) + \sigma_I(\mathbf r^i) \big] + f_{\mathbf r^{j-1}} \lambda_j(r_j) \geq 0 \qquad \forall \mathbf r^j \in \times_{\ell \in [j]} R_\ell, \; \forall j \in [i], \label{eq: FarkasConstraint}\\
    &\boldsymbol{\delta}, \boldsymbol{\sigma} \geq 0. \nonumber
\end{align}
We denote the above model by \eqref{eq: firstFarkas'alternative}.
Note that, in \eqref{eq: FarkasConstraintStrict}, all variables have nonnegative coefficients due to the fact that ${\mathbf{q}''}^*$ is feasible in \eqref{eq: Relaxation}.

We next show that one can find a feasible $(\boldsymbol{\delta},\boldsymbol{\sigma},\boldsymbol{\lambda})$ in \eqref{eq: firstFarkas'alternative} that always satisfies \eqref{eq: FarkasConstraint} with equality. Let $(\boldsymbol{\delta},\boldsymbol{\sigma},\boldsymbol{\lambda})$ denote the feasible solution of the Farkas' alternative that satisfies $\sum_{I \ni j} \sum_{\mathbf r^i} [ \delta_I(\mathbf r^i) + \sigma_I(\mathbf r^i) ] = 1$ and minimizes the sum of all slacks of \eqref{eq: FarkasConstraint}. Note that we can always normalize the value of $\sum_{I \ni j} \sum_{\mathbf r^i} [ \delta_I(\mathbf r^i) + \sigma_I(\mathbf r^i) ]$ to be one, {as any feasible $(\boldsymbol{\delta},\boldsymbol{\sigma},\boldsymbol{\lambda})$ should satisfy $\sum_{I \ni j} \sum_{\mathbf r^i} [ \delta_I(\mathbf r^i) + \sigma_I(\mathbf r^i) ]>0$. If not, as $\boldsymbol{\delta}$ and $\boldsymbol{\sigma}$ are constrained to be nonnegative, we must have $\delta_I(\mathbf r^i) = \sigma_I(\mathbf r^i) = 0$ for all $\mathbf r^i \in \times_{\ell \in [i]} R_\ell$ and $I \subseteq [i]$. Then, constraint \eqref{eq: FarkasConstraint} implies that $\lambda_j(r_j) \geq 0$ for all $r_j \in R_j$ and $j \in [i]$. As the coefficients of $\boldsymbol{\lambda}$ in \eqref{eq: FarkasConstraintStrict} are always nonnegative, we arrive at the conclusion that $(\boldsymbol{\delta},\boldsymbol{\sigma},\boldsymbol{\lambda})$ cannot satisfy \eqref{eq: FarkasConstraintStrict} together with $\boldsymbol{\delta} = \boldsymbol{\sigma}=0$. Hence, any feasible $(\boldsymbol{\delta},\boldsymbol{\sigma},\boldsymbol{\lambda})$ should satisfy $\sum_{I \ni j} \sum_{\mathbf r^i} [ \delta_I(\mathbf r^i) + \sigma_I(\mathbf r^i) ] > 0$, which can be normalized to one.}

Assume to the contrary that the normalized solution, $(\boldsymbol{\delta},\boldsymbol{\sigma},\boldsymbol{\lambda})$, that minimizes the sum of all slacks of \eqref{eq: FarkasConstraint} leaves a strictly positive slack for some $\hat{j} \in [i]$ and $\hat{\mathbf r}^{\hat j} \in \times_{\ell \in [\hat j]} R_\ell$:
\begin{align*}
    \sum_{I \ni \hat j} \sum_{(r_{\hat j+1},\ldots,r_i)} \big[ \delta_I(\hat{\mathbf r}^{\hat j}, r_{\hat j+1},\ldots,r_i) + \sigma_I(\hat{\mathbf r}^{\hat j},r_{\hat j+1},\ldots,r_i) \big] + f_{\mathbf r^{\hat j-1}} \lambda_{\hat j}(r_{\hat j}) > 0.
\end{align*}

\textbf{Case 1 ($\lambda_{\hat j}(r_{\hat j}) > 0$).} In this case, as the other variables, $\boldsymbol{\delta}$ and $\boldsymbol{\sigma}$, are nonnegative, any constraint \eqref{eq: FarkasConstraint} that contains $\lambda_{\hat j}(r_{\hat j})$ must have a strictly positive slack. Hence, as the coefficients of $\boldsymbol \lambda$ in \eqref{eq: FarkasConstraintStrict} are nonnegative, we can strictly decrease $\lambda_{\hat j}(r_{\hat j})$ to arrive at another feasible solution that minimizes the sum of slacks of \eqref{eq: FarkasConstraint}. This would contradict the definition of $(\boldsymbol{\delta},\boldsymbol{\sigma},\boldsymbol{\lambda})$. Hence, we must have $\lambda_{\hat j}(r_{\hat j}) \leq 0$ and $\delta_{\hat I}(\hat{\mathbf r}^{\hat j}, \hat r_{\hat j+1},\ldots,\hat r_i) + \sigma_{\hat I}(\hat{\mathbf r}^{\hat j},\hat r_{\hat j+1},\ldots,\hat r_i) > 0$ for some $(\hat r_{\hat j+1},\ldots,\hat r_i)$ and $\hat I \ni \hat j$.

\textbf{Case 2 ($|\hat I| = 1$ and $\delta_{\hat I}(\hat{\mathbf r}^{\hat j}, \hat r_{\hat j+1},\ldots,\hat r_i)>0$).} In this case, we can construct another feasible solution by strictly decreasing $\delta_{\hat I}(\hat{\mathbf r}^{\hat j}, \hat r_{\hat j+1},\ldots,\hat r_i)$. 
Note that, as $|\hat I| = 1$ and $\hat I \ni \hat j$, we must have $\hat I = \{\hat j\}$, so that $\delta_{\hat I}(\hat{\mathbf r}^{\hat j}, \hat r_{\hat j+1},\ldots,\hat r_i)$ appears in \eqref{eq: FarkasConstraint} only for $\hat j$ and $\hat{\mathbf r}^{\hat j}$, which is assumed to have strictly positive slack.
As any $\boldsymbol{\delta}$ has a nonnegative coefficient in \eqref{eq: FarkasConstraintStrict}, there must exist a small perturbation that maintains feasibility and strictly decreases the sum of all slacks of \eqref{eq: FarkasConstraint}, contradicting the definition of $(\boldsymbol{\delta},\boldsymbol{\sigma},\boldsymbol{\lambda})$.

\textbf{Case 3 ($|\hat I| > 1$ and $\delta_{\hat I}(\hat{\mathbf r}^{\hat j}, \hat r_{\hat j+1},\ldots,\hat r_i)>0$).}  In this case, we can define another feasible solution, $(\boldsymbol{\delta}',\boldsymbol{\sigma},\boldsymbol{\lambda})$ as follows:
\begin{align*}
    \delta_J'({\mathbf{r}}^i) &= \begin{cases}
        \delta_J({\mathbf{r}}^i) - \varepsilon &\text{if } J = \hat I, \text{ and } {\mathbf{r}}^i = \hat{\mathbf r}^i,\\
        \delta_J({\mathbf{r}}^i) + \varepsilon &\text{if } J = \hat I \setminus \{\hat j\}, \text{ and } {\mathbf{r}}^i = \hat{\mathbf r}^i,\\
        \delta_J({\mathbf{r}}^i) &\text{otherwise,}
    \end{cases}
    &\forall J \in [i], \forall \mathbf r^i {\in \times_{\ell \in [i]} R_\ell},
\end{align*}
for some $\varepsilon>0$.
As $(\boldsymbol{\delta},\boldsymbol{\sigma},\boldsymbol{\lambda})$ satisfies \eqref{eq: FarkasConstraintStrict} with a strict inequality, choosing $\varepsilon$ small enough will ensure that $(\boldsymbol{\delta}',\boldsymbol{\sigma},\boldsymbol{\lambda})$ also satisfies \eqref{eq: FarkasConstraintStrict}. As $\delta_{\hat I}(\hat{\mathbf r}^{\hat j}, \hat r_{\hat j+1},\ldots,\hat r_i)>0$, the nonnegativity of $\boldsymbol{\delta}'$ can also be ensured by choosing $\varepsilon>0$ accordingly. Lastly, we show that $(\boldsymbol{\delta}',\boldsymbol{\sigma},\boldsymbol{\lambda})$ also satisfies \eqref{eq: FarkasConstraint} for all $k \in [i]$ and ${\mathbf{r}}^k \in \times_{\ell \leq k} R_\ell$. Notice that the left-hand side of \eqref{eq: FarkasConstraint} is not affected by $\varepsilon$ if $k \not \in \hat I$ or $\mathbf r^k \neq \hat{\mathbf{r}}^k$. For $k \in \hat I \setminus \{\hat j\}$ and $\mathbf r^k = \hat{\mathbf{r}}^k$, the perturbation, $\varepsilon$, cancels out. For $\hat j$ and $\hat{\mathbf r}^{\hat j}$, we end up with a strictly decreased slack. Hence, we arrive at a contradiction. For the case of $\sigma_{\hat I}(\hat{\mathbf r}^{\hat j},\hat r_{\hat j+1},\ldots,\hat r_i)>0$, one can replicate the arguments in Cases 2 and 3 by defining another feasible solution $(\boldsymbol{\delta},\boldsymbol{\sigma}',\boldsymbol{\lambda})$ in a similar way. Hence, we establish that one can find a normalized feasible solution of \eqref{eq: firstFarkas'alternative} that satisfies \eqref{eq: FarkasConstraint} always with an equality.

Then, we can write the following:
\begin{multline*}
    \sum_{I \subseteq [i]} \sum_{\mathbf r^i {\in \times_{\ell \in [i]} R_\ell}} \Bigg[ \Big( \delta_I(\mathbf r^i) + \sigma_I(\mathbf r^i) \Big) \sum_{j \in I} q_j''^*(\mathbf r^j) \Bigg] \\
    = \sum_{j \in [i]} \sum_{\mathbf r^j {\in \times_{\ell \in [j]} R_\ell}}q_j''^*(\mathbf r^j) \Bigg[ \sum_{I \ni j} \sum_{(r_{j+1},\ldots,r_i)} \delta_I(\mathbf r^i) + \sigma_I(\mathbf r^i) \Bigg].
\end{multline*}
Using the fact that $(\boldsymbol{\delta},\boldsymbol{\sigma},\boldsymbol{\lambda})$ satisfies \eqref{eq: FarkasConstraint} always with equality, we deduce that:
\begin{align*}
    \sum_{I \subseteq [i]} \sum_{\mathbf r^i} \Bigg[ \Big( \delta_I(\mathbf r^i) + \sigma_I(\mathbf r^i) \Big) \sum_{j \in I} q_j''^*(\mathbf r^j) \Bigg] &= \sum_{j \in [i]} \sum_{\mathbf r^j {\in \times_{\ell \in [j]} R_\ell}}q_j''^*(\mathbf r^j)\Bigg[ - f_{r^{j-1}} \lambda_j(r_j) \Bigg]\\
    &= - \sum_{j \in [i]} \sum_{r_j \in R_j} \lambda_j(r_j) \sum_{\mathbf r^{j-1}} f_{\mathbf r^{j-1}}q_j''^*(\mathbf r^j) \\
    & = - \sum_{j \in [i]} \sum_{r_j \in R_j} \lambda_j(r_j) \mathbb{E}_{\mathbf r^{j-1}}[q_j''^*(\mathbf r_j)].
\end{align*}
Then, the first constraint, \eqref{eq: FarkasConstraintStrict}, of the Farkas' alternative implies:
\begin{multline}\label{eq: FarkasConstraintStrictTwo}
    \sum_{I \subseteq[i]} \sum_{\mathbf r^i {\in \times_{\ell \in [i]} R_\ell}} \Bigg[ \delta_I(\mathbf r^i) g(I \cup \{i+1\},\mathbf r^I) + \sigma_I(\mathbf r^i) g(I,\mathbf r^I) \Bigg] \\
    + \sum_{j \in [i]} \sum_{r_j \in R_j} \lambda_j(r_j) 0.5{W_j^*(r_j)} <0.
\end{multline}
That is, there exists $(\boldsymbol{\delta},\boldsymbol{\sigma},\boldsymbol{\lambda})$ that has nonnegative $(\boldsymbol{\delta},\boldsymbol{\sigma})$ and satisfies \eqref{eq: FarkasConstraintStrictTwo} together with \eqref{eq: FarkasConstraint}. Due to Farkas' Lemma, this latter system, (\eqref{eq: FarkasConstraintStrictTwo}, \eqref{eq: FarkasConstraint}, and nonnegative $(\boldsymbol{\delta},\boldsymbol{\sigma})$), implies the infeasibility of the following one:
\begin{align*}
	&\sum_{j \in I} q_j(\mathbf r^j) \leq g(I \cup \{i+1\},\mathbf r^I) &\forall \mathbf r^i {\in \times_{\ell \in [i]} R_\ell}, \forall I \subseteq[i],\\
    &\sum_{j \in I} q_j(\mathbf r^j) \leq g(I,\mathbf r^I) &\forall \mathbf r^i {\in \times_{\ell \in [i]} R_\ell}, \forall I \subseteq[i],  \\
	&\mathbb{E}_{\mathbf r^{j-1}}[ q_j(\mathbf r^j)] = 0.5{W_j^*(r_j)} &\forall r_j {\in R_j}, \; \forall j \in [i], \\ 
	&{{\mathbf{q}}} \geq 0, \nonumber
\end{align*}
which means that \eqref{eq: MainProblem} is infeasible. This is a contradiction; therefore, \eqref{eq: LowerBound} and \eqref{eq: Relaxation} must have the same optimal solution.
\end{proof}

Through Lemmas \ref{lem: EquivalenceResult}, \ref{lem: LowerBoundModel}, and \ref{lem: equivalentLowerBound}, we prove that a relatively easier model, \eqref{eq: Relaxation}, can be solved to obtain a lower bound for \eqref{eq: MainProblem}. Our main result in this section leverages this fact in order to prove the prophet inequality of $1/2$ for on-line Polymatroids.

\begin{thrm}\label{thrm: onlinePolymatroidMainResult}
    If $g$ is a nonnegative, nondecreasing, and on-line submodular set function, then the optimal on-line strategy achieves at least $\frac{1}{2}$ of the prophet's value.
\end{thrm}
\begin{proof}
We will use induction to
prove our result.
Our base case is $i = 1$. Let $q_1(r_1) = 0.5W_1^*(r_1)$ for all $r_1 \in R_1$ and $z(r_1) = 0.5g(\{2\})$. Then, the first constraint of \eqref{eq: MainProblem} is satisfied since, for any $I \subseteq [i]$, 
$$
\sum_{j \in I} q_j(\mathbf r^j) + z(r_1) \leq 0.5{g(I,\mathbf r^I)} + 0.5g(\{2\}) \leq g(I \cup \{2\},\mathbf r^I),
$$
where the inequalities follow from the following facts. First, $W_1^*(r_1) \leq g(I,r_1)$ for any $r_1 \in R_1$ and any $I \ni 1$. Second, as $g$ is nondecreasing, $g(\{2\}) \leq g(I \cup \{2\},\mathbf r^I)$ and ${g(I,\mathbf r^I)} \leq g(I \cup \{2\},\mathbf r^I)$ for any $I \subseteq [i]$. It is straightforward to show that $(q_1,z)$ given above also satisfies the remaining constraints of \eqref{eq: MainProblem}.

Our induction hypothesis is as follows. There exists $i \geq 2$ such that we can find an on-line solution $\mathbf{q}$ feasible in \eqref{eq: MainProblem}, for any nonnegative, nondecreasing, and on-line submodular function that allocates $0.5W_j^*(r_j)$ in expectation to all $r_j \in R_j$ and $j \in [i]$. Next, we consider a specific on-line submodular function $g$, which is non-negative and non-decreasing, and prove that it is also possible to allocate $0.5W_{i+1}^*(r_{i+1})$ in expectation to all $r_{i+1} \in R_{i+1}$. To this end, we consider the following model, which is feasible thanks to our induction hypothesis and yields a lower bound for $h_{i+1}(0.5{\mathbf W^i}^*)$ thanks to Lemmas \ref{lem: EquivalenceResult}, \ref{lem: LowerBoundModel}, and \ref{lem: equivalentLowerBound}:
\begin{align}
    \max_{\mathbf{q}', \mathbf{q}''}\; &\mathbb{E}_{\mathbf r^i} [g(\{i+1\}) - \sum_{j \in [i]} q_j'(r^j)]  \tag{\ref{eq: Relaxation}}\\
	\text{s.t.}\; &\sum_{j \in I} \Big[q_j'(\mathbf r^j) + q_j''(\mathbf r^j) \Big] \leq g(I \cup \{i+1\},\mathbf r^I) &\forall \mathbf r^i {\in \times_{\ell \in [i]} R_\ell}, \forall I \subseteq[i], \nonumber \\
    &\sum_{j \in I} q_j''(\mathbf r^j) \leq g(I \cup \{i+1\},\mathbf r^I) - g(\{i+1\}) &\forall \mathbf r^i {\in \times_{\ell \in [i]} R_\ell}, \forall I \subseteq[i], \tag{\ref{eq: q''constraint}}\\
	&\mathbb{E}_{\mathbf r^{j-1}} \Big[q_j'(\mathbf r^j) + q_j''(\mathbf r^j) \Big] = 0.5{W_j^*(r_j)} &\forall r_j {\in R_j}, \; \forall j \in [i], \nonumber \\ 
	&{{\mathbf{q}}}', {{\mathbf{q}}}'' \geq 0. \nonumber
\end{align}

Let us define another set function as follows: $\hat g(I,\mathbf r^I) = g(I \cup \{i+1\},\mathbf r^I) -g(\{i+1\})$ for all $I \subseteq [i]$ and $\mathbf r^i \in \times_{\ell \in [i]} R_\ell$. It is easy to verify that $\hat g$ shares the properties of $g$, $i.e.$, it is nonnegative, nondecreasing, and on-line submodular. We can consider the off-line problem subject to $\hat{g}$ and retrieve expected allocation values, denoted by $\hat{W}_j(r_j)$ for all $r_j \in R_j$ and $j \in [i]$. If we compare these values to the original expected allocation values, we can show that the following must hold:
\begin{align*}
    &\hat W_j(r_j) \leq W_j^*(r_j) &\forall r_j \in R_j, \forall j \in [i],\\
    &\sum_{j \in [i]}\sum_{r_j \in R_j} f_{r_j} \hat W_j(r_j) =g( [i+1]) -g(\{i+1\}), 
\end{align*}
where the inequality follows from $\hat w_j(\mathbf{r}^i) \leq w_j^*(\mathbf{r}^i)$ for all $\mathbf{r}^i {\in \times_{\ell \in [i]} R_\ell}$, and $j \in [i]$, and the equality from the optimality of $\hat w$ (all rewards are non-negative).

Due to our induction hypothesis, there must exist a feasible on-line solution, $\hat{\mathbf{q}}$, contained in on-line $\hat g$ that allocates $0.5\hat W_j(r_j)$ in expectation to all $r_j \in R_j$ and $j \in [i]$. Then, in the formulation \eqref{eq: Relaxation}, we can use $\hat{\mathbf{q}}$ as $\mathbf{q}''$, which leaves only $0.5W_j'(r_j) = 0.5W_j^*(r_j) - 0.5\hat W_j(r_j) \geq 0$ for all $r_j \in R_j$ and $j \in [i]$ to be allocated in expectation by $\mathbf q'$. Regarding $\mathbf W'$ values, we know the following:
\begin{align}
    &W_j'(r_j) \leq {g(\{i+1\})}, &\forall r_j \in R_j, \forall j \in [i], \nonumber\\
    &\sum_{j \in [i]}\sum_{r_j \in R_j} f_{r_j} W_j'(r_j) \leq g(\{i+1\}) \label{eq: hatW_upperBound}, 
\end{align}
where both inequalities follow from $ w_j'(\mathbf{r}^i) \leq g(\{i+1\})$ for a specific choice of $\mathbf{w}'$, which we prove next. We consider both optimal $\mathbf{w}^*$ and $\hat{\mathbf{w}}$ to be constructed via Greedy algorithm under the same tie breaking rules and let $\mathbf w_j' = \mathbf w_j^* - \hat{\mathbf w}_j$. That is, for any $\mathbf{r}^i$, for some $j \in [n]$ with $r_j \in \arg \max_{\ell \in [i]} r_\ell$, we have $w_j^*(\mathbf{r}^i) = g(\{j\}, r_j)$, and $\hat w_j(\mathbf{r}^i) = g(\{j,i+1\},{r}_j) -g(\{i+1\})$, which implies
\begin{align*}
    w_j'(\mathbf{r}^i) =  g(\{i+1\}) + g(\{j\}, r_j) - g(\{j,i+1\},{r}_j) \leq g(\{i+1\}).
\end{align*}
For the case of $r_j \not \in \arg \max_{\ell \in [i]} r_\ell$, we would have $w_j^*(\mathbf{r}^i) = g(I, \mathbf r^I) - g(I \setminus \{j\}, \mathbf r^{I \setminus \{j\}})$, and $\hat w_j(\mathbf{r}^i) = g(I \cup \{i+1\},\mathbf{r}^I) - g(I \cup \{i+1\} \setminus \{j\},\mathbf{r}^{I \setminus \{j\}})$ for some $I \ni j$. Then, we would have
\begin{align*}
    w_j'(\mathbf{r}^i) &=  g(I, \mathbf r^I) - g(I \setminus \{j\}, \mathbf r^{I \setminus \{j\}}) - g(I \cup \{i+1\},\mathbf{r}^I) + g(I \cup \{i+1\} \setminus \{j\},\mathbf{r}^{I \setminus \{j\}})\\
    &\leq g(I \cup \{i+1\} \setminus \{j\},\mathbf{r}^{I \setminus \{j\}}) - g(I \setminus \{j\}, \mathbf r^{I \setminus \{j\}}) \leq g(\{i+1\}).
\end{align*}
Notice that in any case $\mathbf{w}'$ is nonnegative, and the same line of arguments would also imply $\sum_{j \in [n]} w_j'(\mathbf{r}^i) \leq g(\{i+1\})$ for any $\mathbf{r}^i$.

As $\mathbf{w}'$ always satisfies a knapsack constraint, due to Proposition \ref{prop:oneConstrained}, there exists a feasible $\mathbf q'$ that satisfies $\sum_{j \in I} q_j'(\mathbf r^j) \leq g(\{i+1\})$ for all $\mathbf r^i {\in \times_{\ell \in [i]} R_\ell}$ and allocates $0.5{\mathbf{W}'}^{i}$ in expectation to all $r_j \in R_j$ for all $j \in [i]$. 

We now show that $(\mathbf{q}',\mathbf{q}'')$ is feasible in \eqref{eq: Relaxation}:
\begin{align*}
    \sum_{j \in I} \Big[q_j'(\mathbf r^j) + q_j''(\mathbf r^j) \Big] &\leq g(\{i+1\}) + \Big[ g(I \cup \{i+1\}, \mathbf r^I) - g(\{i+1\}) \Big] \\
    &\leq  g(I \cup \{i+1\}, \mathbf r^I)  \qquad \forall \mathbf r^i {\in \times_{\ell \in [i]} R_\ell}, \forall I \subseteq [i],
\end{align*}
where the first inequality follows from the facts that $\mathbf{q}'$ is nonnegative and lies in the knapsack $g(\{i+1\})$, and $\mathbf{q}''=\hat{\mathbf{q}}$. The other constraints of \eqref{eq: Relaxation} are satisfied due to the construction of $(\mathbf{q}',\mathbf{q}'')$.

Lastly, as $(\mathbf{q}',\mathbf{q}'')$ is a feasible solution of \eqref{eq: Relaxation}, its  optimal objective value is at least:
\begin{align*}
    \mathbb{E}_{\mathbf r^i} [g(\{i+1\}) - \sum_{j \in I} q_j'(\mathbf r^j)] &= g(\{i+1\}) - \sum_{j \in [i]}\sum_{r_j \in R_j} f_{r_j} 0.5 { W_j'(r_j)} \\
    &\geq 0.5{g(\{i+1\})},
\end{align*}
where the inequality follows from \eqref{eq: hatW_upperBound}. 
\end{proof}

\subsection{Correlated Rewards}\label{section:Correlated}

In this section, we drop the assumption of independent rewards and derive the $\frac{1}{n}$-prophet inequality from \citet{hill1992survey} for a general set of constraints. This is best possible, see \citet{hill1981ratio}.\footnote{Stronger bounds can be obtained if one restricts the form that correlations can take, see, for example, \citet{ImmorlicaSinglaWaggoner2020}.}

The prophet's problem is:
\begin{align}
    \max_{\mathbf{w}} \;& \mathbb{E}_{\mathbf{r}^n}[\sum_{\mathbf{r}^n \in \times_{\ell \in [n]} R_\ell} \sum_{j \in [n]}  r_j w_j(\mathbf{r}^n)] \label{eq: correlatedProphetProblem} \tag{$\mathcal{P}_C$}\\
    \text{s.t.}\; & \sum_{j \in [n]} a_{kj}w_j(\mathbf{r}^n) \leq b_k &\forall \mathbf r^n {\in \times_{\ell \in [n]} R_\ell}, \forall k \in [K], \nonumber\\
    &\mathbf{w} \geq 0, \nonumber
\end{align}
where the constraint coefficients, $\{a_{kj}\}_{k \in [K],j \in [n]}$, and the constraint bounds, $\{b_k\}_{k \in [K]}$, are all nonnegative. Then, due to dependence, the interim allocations of optimal allocation $\mathbf{w}^*$ can be written as:
\begin{align*}
    &W_j^*(r_j) = \mathbb{E}_{\mathbf{r}^n \,|\,r_j}[w_j^*(\mathbf{r}^n)] = \sum_{\mathbf{r}^{-j}} \frac{f(r_j,\mathbf{r}^{-j})}{f(r_j)} w_j^*(r_j,\mathbf{r}^{-j}) &\forall r_j \in R_j, \forall j \in [n].
\end{align*}
The main result of this section proves that there exists a feasible on-line strategy that allocates ${W}_j^*(r_j)/n$ to each $r_j \in R_j$ for all $j \in [n]$.
 
\begin{thrm}
    If \eqref{eq: correlatedProphetProblem} is feasible, and $\{a_{kj}\}_{k \in [K],j \in [n]}$ and $\{b_k\}_{k \in [K]}$ are nonnegative, then the optimal on-line strategy achieves at least $\frac{1}{n}$ of the prophet's value.
\end{thrm}
\begin{proof}
We will prove the feasibility of the following system of inequalities.
\begin{align}
	&\sum_{j \in [n]}a_{kj}w_j(\mathbf{r}^n)  \leq b_k  &\forall \mathbf r^n {\in \times_{\ell \in [n]} R_\ell}, \forall k \in [K], \label{eq: correlatedOnlineConstraint1}\\
	&w_j(\mathbf{r}^n) = q_j(\mathbf r^j) &\forall \mathbf r^n {\in \times_{\ell \in [n]} R_\ell}, \forall j \in [n], \label{eq: correlatedOnlineConstraint2}\\
	&\mathbb{E}_{\mathbf{r}^n \,|\,r_j}[w_j(\mathbf{r}^n)] \geq \frac{1}{n}W_j^*(r_j) &\forall r_j \in R_j, \forall j \in [n], \label{eq: correlatedOnlineConstraint3}\\
	&\mathbf{w}, \mathbf{q} \geq 0. \nonumber
\end{align}

We work with the ex-post allocation variables for the off-line problem, $\mathbf w$, in (\ref{eq: correlatedOnlineConstraint1})-(\ref{eq: correlatedOnlineConstraint3}) to avoid any confusion arising from the dependence between rewards. The third constraint is the non-anticipativity ones that ensure that for any $j \in [n]$, $w_j(\mathbf r^n)$ variables have the same value for all $(r_{j+1}, \ldots, r_n)$. In other words, $\mathbf w$ agrees with some on-line ex-post allocation rule $\mathbf q$. 

Given \eqref{eq: correlatedOnlineConstraint1} we deduce the following constraints on $\mathbf{W}^n$:
\begin{align}
    a_{kj}w_j^*(\mathbf r^n) &\leq b_k & \forall \mathbf r^n {\in \times_{\ell \in [n]} R_\ell}, \forall j \in [n], \forall k \in [K], \nonumber\\
    a_{kj} W_j^*(r_j)= \mathbb{E}_{\mathbf{r}^n \,|\,r_j}[a_{kj}w_j^*(\mathbf{r}^n)] &\leq b_k &\forall r_j \in R_j, \forall j \in [n], \forall k \in [K]. \label{eq: correlated_W_bound}
\end{align}
The first inequality follows because $a_{k \ell} \geq 0$ for all $\ell \in [n]$ and $k \in [K]$, as well as $\mathbf{w}^* \geq 0$. The second inequality follows by taking the conditional expectation of \eqref{eq: correlatedOnlineConstraint1}. 

We conclude the proof by verifying that the solution $w_j'(\mathbf{r}^n) = q_j'(\mathbf{r}^j) = W_j^*(r_j)/n$ satisfies the system (\ref{eq: correlatedOnlineConstraint1}-\ref{eq: correlatedOnlineConstraint3}). The proposed solution, $(\mathbf{w}',\mathbf{q}')$, by definition, is nonnegative and satisfies the constraints \eqref{eq: correlatedOnlineConstraint2} and \eqref{eq: correlatedOnlineConstraint3}. Constraint \eqref{eq: correlatedOnlineConstraint1} is also satisfied because:
\begin{align*}
    \sum_{j \in [n]} a_{kj} w_j'(\mathbf{r}^n) = \sum_{j \in [n]} a_{kj} \frac{W_j^*(r_j)}{n} \leq \sum_{j \in [n]} \frac{b_k}{n} = b_k,
\end{align*}
where the inequality follows from \eqref{eq: correlated_W_bound}.
\end{proof}

\section{Concluding Remarks}
We proposed a linear programming approach based on interim allocation variables to systematically derive various known and new prophet inequalities. This approach can be further extended by considering non-uniform scalings of the optimal off-line interim variables which will be useful for the case of selection problems with random horizons (see \citet{GiambartolomeiMallmannTrennSaona2025}).

\bibliography{arxivbib}


\end{document}